\documentclass[11pt]{article}
\usepackage{amssymb,amsmath,amsthm,verbatim, wrapfig, graphicx, caption,cite,tikz,pgf,mathtools}
\usepackage{amssymb}
\usepackage{amsmath}
\usepackage{amsfonts}
\usepackage{amsthm}

\newcommand{\brac}[1]{\ensuremath{\left( #1 \right)}}

\def\R{\ensuremath{I\!\!R}}

\DeclareMathOperator{\im}{im}

\DeclareMathOperator{\sign}{sign}

\DeclareMathOperator{\inte}{int}

\newtheorem{theorem}{Theorem}
\newtheorem{definition}{Definition}

\newtheorem{corollary}{Corollary}
\newtheorem*{example*}{Example}
\newtheorem{lemma}{Lemma}

\newtheorem{remark}{Remark}
\newtheorem{assumption}{Assumption}

\definecolor{dblue}{rgb}{0.0,0.0,0.68}

\definecolor{dred}{rgb}{0.75,0,0}
\def\dred{\textcolor{dred}}
\definecolor{dgreen}{rgb}{0,0.8,0}

\title{Graph-theoretic analysis of multistationarity using  degree theory}

\author{Carsten Conradi \thanks{ Max Planck Institute Dynamics of Complex Technical Systems, Magdeburg, Germany. Email: \texttt{conradi@mpi-magdeburg.mpg.de}} \ and 
Maya Mincheva\thanks{Corresponding author. Department of Mathematical Sciences, 
Northern Illinois University, DeKalb,  IL 60115, USA.
E-mail: \texttt{mincheva@math.niu.edu}}}

\begin{document}
\maketitle

\begin{abstract}
Biochemical mechanisms  with mass action kinetics are often modeled by  systems of  polynomial  differential equations (DE).
Determining  directly if the DE system has multiple equilibria (multistationarity) is difficult for realistic systems, since they are large, nonlinear and contain many unknown parameters. 
Mass action biochemical mechanisms can be represented by a directed bipartite graph with species and reaction nodes.
Graph-theoretic methods can  then 
be used to assess the potential of a given biochemical mechanism for
multistationarity by identifying structures in the bipartite graph 
referred to as critical fragments. In this article we present a
graph-theoretic method for conservative biochemical mechanisms
characterized by bounded species concentrations, which makes the use of degree 
theory arguments  possible. We illustrate the results with an example
of a MAPK network.

\end{abstract}

\section{Introduction}\label{intro}
Biochemical mechanisms    of  chemical species  and elementary reactions  are often modeled by  differential equations (DE) systems with the species concentrations as variables.
Multistability, the existence of multiple  stable positive equilibria
(for some choice of parameter values) is   ubiquitous in models of
biochemical mechanisms,  such as  cell decision
\cite{markevich2004signaling,ozbudak2004multistability}. And 
multistationarity,  the existence of multiple  positive equilibria  is
necessary for multistability or a biological switch, a term used in
the biological literature.

The  models  in this work  will be taken with mass action kinetics resulting in a polynomial right-hand side  of the DE system.
The DE system models of the biochemical mechanisms  of interest are
typically high-dimensional, nonlinear and contain many unknown
parameters (rate constants and total concentrations). Thus determining
parameter values such that multiple equilibria   can  be found  by
solving numerically  large  nonlinear polynomial systems is difficult,
if not impossible.  On the other hand solving a nonlinear polynomial
system with unknown coefficients directly  using methods from algebraic
geometry has its limitations \cite{JoshiShiu2012}. Therefore other
methods and approaches  such as graph-theoretic 
are being  developed to answer the question of
the existence of multistationarity more easily.

A biochemical  mechanism  with mass action kinetics  can be
represented as a directed bipartite  graph, which is a graph with two
non-intersecting  sets of nodes representing species and reactions,  and
directed edges starting at a species (reaction) node and ending at a
reaction (species) node. Graph-theoretic methods can be used  to
identify  structures referred to as critical fragments  that are
necessary   for the existence of   multistationarity
\cite{Mincheva2007}.

Graph-theoretic methods have been  used to determine the potential of various biochemical mechanisms for multistationarity \cite{BanajiCraciun2010,Craciun2006,Mincheva2007,VolpertIvanova}.
Many of these methods use the one-to-one correspondence  between structures in the graph (fragments \cite{Mincheva2007,VolpertIvanova}  or cycle structures \cite{Craciun2006})
and the summands in the determinant of the Jacobian of the right-hand side  of  the DE system. However, many models 
have conservation relations with positive coefficients  of the species concentrations.
The existence of conservation relations results in a non-full rank  Jacobian.  This leads to  considering a coefficient of the characteristic polynomial of the Jacobian different from the constant coefficient and its sign when studying multistationarity. Here we study  
 conservative biochemical mechanisms  where all species concentrations participate in at least one conservation relation. This means that all species concentrations are bounded from above and degree theory \cite{Deimling2010} can be used to study the number of equilibria of the DE  model.
So far degree theory has been used to study multistationarity in biochemical mechanism  models, for example,  in \cite{ConradiMincheva2014,CraciunHelton2008,DeSontag2007} and  graph-theoretic methods have been developed in \cite{BanajiCraciun2010,Craciun2006,Mincheva2007,VolpertIvanova}. Here we   
combine both approaches to develop  a graph-theoretic  method for multistationarity in  a conservative  biochemical mechanism DE model.

This article is organized as follows. Sec.~\ref{prel} provides an
introduction to  conservative biochemical mechanisms with mass action
kinetics and their properties. 
 In Sec.~\ref{sec:dyn_constr} we discuss consequences
  of the well-known fact that solutions of the DE systems under study
  are confined to affine linear subspaces defined by these conservation
  relations.
In Sec.~\ref{sec:jac_ma} the Jacobian of the original  DE system, its
parametrization and  the determinant of  the Jacobian on the level
sets  is given. In Sec.~\ref{sec:degree_g} the degree of a nonlinear
function   and some of its properties related to  a DE system's
right-hand side  on a given  level set  is introduced. In
Sec.~\ref{sec:bip_gr} the bipartite graph of a biochemical mechanism
with mass action kinetics is introduced. The main result in
Sec.~\ref{sec:multi} (Theorem ~\ref{thm2} and Corollary
\ref{cor:cor4}) gives a necessary condition for multistationarity for
conservative biochemical mechanism models. An example of a MAPK
network  model studied for multistationarity  in \cite{fein-043} is
presented in the same section.

\section{Preliminaries}\label{prel}
A (bio)chemical mechanism with 
$n$ species $A_i$, $i=1,\ldots, n$, and $m$ elementary
reactions is represented as 
\begin{equation}\label{cr}
\sum_{i=1}^{n} \alpha_{ij} A_{i} \xrightarrow{k_{j}} \sum_{i=1}^{n}   
\beta_{ij} A_{i} ,\quad j=1 \dotsc m, 
\end{equation}
where   $k_j>0$, $j=1, \ldots , m$ are the  {\it rate constants}.
The constants
$\alpha_{ij} \geq 0$ and $\beta_{ij} \geq 0$ are small integers called 
{\it stoichiometric coefficients}  that account for the
number of molecules of species $A_i$ participating in the $j^\mathrm{th}$ elementary 
reaction in (\ref{cr}). An example of a  chemical mechanism is given below 
\begin{equation}
        \begin{array}{lrcl}
          &  A_2 + A_3 &\xrightarrow{k_1} & 2 A_1 , \\
           & A_3 & \xrightarrow{k_2} & A_1, \\ 
            & A_1 & \xrightarrow{k_3} & A_3, \\ 
           & A_2  & \xrightarrow{k_4} &  A_1 \\ 
             & A_1  & \xrightarrow{k_5} &  A_2. \\ 
        \end{array}
        \label{eq:br_cyc}
      \end{equation}

\begin{assumption}\label{ass1}
A true reaction is a reaction different from an
  inflow reaction $A_i \to \emptyset$ or an outflow reaction
  $\emptyset \to A_i$. An autocatalytic reaction is a  reaction  of
  the form $s_1A_i +...\to s_2 A_i +\ldots $ where $0<s_1<s_2$. We assume
  that every species in (\ref{cr}) is consumed and produced  in at
  least one true non-autocatalytic elementary
  reaction.
\end{assumption} 
The chemical mechanism in (\ref{eq:br_cyc}) satisfies the above assumption.

We will denote by $x=(x_1, \ldots, x_n)$ the vector of concentrations
$x_i$ of species $A_i$ and by $k=(k_1, \ldots, k_m)$ the vector of
rate constants. If for $y \in \mathbb{R}^n$,  $y_i\geq 0$ ($y_i>0$)
for all $i$ we will write $y\geq 0$  $(y>0)$. Since each $x_i \geq 0$
as  a concentration,  we have   $x\geq 0 $.
Similarly $k>0$ since each rate constant $k_j>0$.

If  {\it mass action kinetics}  is used for the mechanism (\ref{cr}), then the corresponding  {\it rate functions} are 
\begin{equation}\label{rf}
v_j (k,x) = k_j x_{1}^{\alpha_{1j}} \ldots  x_{n}^{\alpha_{nj}},  \quad \quad j =1, \ldots , m.
\end{equation}
The vector of rate functions will be denoted as  $v(k,x)  =(v_1 (k,x), \ldots, v_m (k,x))$ where $v\geq 0$.

The differential equation (DE) model of a mass-action mechanism such as  (\ref{cr}) can  be  written as 
\begin{equation}\label{crm}
\dot{x} (t) = N v(k,x) =f(k,x) =f(v,x)
\end{equation}
where $$N_{ij}= \beta_{ij} -\alpha_{ij}$$ 
are the entries of  the stoichiometric matrix $N$ with dimension  $(n
\times m)$ and  $v(k,x)$ is the vector of rate functions (\ref{rf}). 
  In what follows we will at times interpret the right hand side of
  (\ref{crm}) as a function of the rate constants $k$ and the species
  concentrations $x$ and at other times as a function of the reaction
  rates $v$ and the concentrations $x$, depending on the situation.

  For the system (\ref{eq:br_cyc}) we obtain the stoichiometric matrix
  \begin{displaymath}
    N = \left[
      \begin{array}{ccccc}
        \phantom{-}2 & \phantom{-}1 & -1 & \phantom{-}1 & -1 \\
        -1 & \phantom{-}0 & \phantom{-}0 & -1 & \phantom{-}1 \\
        -1 & -1 & \phantom{-}1 & \phantom{-}0 & \phantom{-}0 
      \end{array}
    \right] \text{ and the rate functions }  v(k,x) =
    (k_1x_2x_3,k_2x_3,k_3x_1,k_4x_2,k_5x_1)^T\ .
  \end{displaymath}

The equations of the system (\ref{crm}) can be written componentwise as 
 \begin{equation}\label{eq:crm_com}
 \dot{x}_i (t) =\sum_{j=1}^m N_{ij} v_j (k,x) , \quad i=1, \ldots , n.
 \end{equation}
The model equations of  the mechanism  (\ref{eq:br_cyc}) are  given below
\begin{equation}\begin{aligned} \label{eq:de_cyc}
\dot{x}_1 &= 2k_1x_2x_3  +k_2 x_3  - k_3 x_1 +k_4x_2 -k_5x_1 =2v_1+v_2-v_3+v_4-v_5,\\
\dot{x}_{2} &=  -k_1 x_2 x_3  -k_4x_2 +k_5 x_1 = -v_1-v_4+v_5,\\
\dot{x}_{3} &= -k_1 x_2 x_3 - k_2 x_{3}  +k_3 x_1= -v_1-v_2+v_3. \\  
\end{aligned}\end{equation}

Initially $x(0) =x_0 \geq 0$, and we will denote a solution $x(t)$ of
(\ref{crm}) with initial condition $x_0$ as $x(t,x_0)$. 

\section{The dynamics on the  level sets $\omega_{c_0}$}\label{sec:dyn_constr}

Let the stoichiometric matrix $N$ have  rank $r$. 
 Suppose that  at least one  solution $\lambda =(\lambda_1,\ldots ,\lambda_n) $ of the system
 \begin{equation}\label{eq:zero-stoic}
  \sum_{i=1}^n  \lambda_i N_{ij} =0, \quad j=1,\ldots, m
  \end{equation}
 exists. Then we have $(n-r)$ mass conservation laws
 \begin{equation}\label{eq:mcr} 
 \sum_{i=1}^n  \lambda_i x_i=  \sum_{i=1}^n  \lambda_i x_{i} (0).
 \end{equation}
If $\lambda_i>0$ for all $i$ in at least one solution $\lambda$,
then by (\ref{eq:mcr}) it follows that all species concentrations
$x_i$ are {\it conserved}, i.e.,  $0\leq x_i (t) \leq M$ for all $i$
where $M >0$. A biochemical mechanism  (\ref{cr}) with  mass-conserved
species concentrations  will be called  {\it conservative biochemical
  mechanism} \cite{CraciunHelton2008}.

\begin{assumption}\label{ass2}
  Here we study only biochemical mechanisms where the system
    \begin{displaymath}
      \lambda^T\, N = 0, \lambda>0
    \end{displaymath}
    has a solution. As outlined above, in this situation all species
    concentrations are bounded.
  
 \end{assumption}

The rank of the stoichiometry matrix $N$ of the system
(\ref{eq:de_cyc}) equals $2$. 
  The left kernel of N is spanned by the vector $\lambda^T =
  (1,1,1)$. Hence the network is conservative and there exists one
  conservation relation, $x_1+x_2+x_3=const.$

We can rewrite (\ref{eq:mcr}) in a matrix form as
 \begin{equation}\label{eq:crel}
  W^T\, x(t,x_0) \equiv W^T\, x_0 =c_0
\end{equation}
where $W$ is a full rank $n\times (n-r)$ matrix  whose columns span
$\ker \brac{N^T}$.

  In what follows we interpret the entries of the $(n-r)$ dimensional vector $c_0$ as
  additional parameters and study the dynamics of the system
  (\ref{crm}) on the level sets

\begin{equation}
  \label{eq:def_om}
  \omega_{c_0} = \left\{  x \geq 0 \; | \; W^T x = c_0 \right\}.
\end{equation}
  This is motivated by the observation that the sets $\omega_{c_0}$
  are invariant under the dynamics of (\ref{crm}).

\begin{lemma}[$\omega_{c_0}$ convex, compact and forward invariant]
  \label{lem:Om_forward_invariant}
  The set $\omega_{c_0}$ is convex, compact and forward invariant. 
\end{lemma}

  The  proof of  Lemma~\ref{lem:Om_forward_invariant}  
  is available in \cite{ConradiMincheva2014}.

 To study the dynamics of system (\ref{crm}) on invariant sets 
  $\omega_{c_0}$ we let $S\in\R^{n\times r}$ be the matrix of full 
  column rank whose columns are an orthonormal basis of $\im(N)$ and
  we let the matrix $Z\in\R^{n\times (n-r)}$ be the matrix of full
  column rank whose columns are the orthonormal basis of
  $\im(N)^\perp\equiv\ker\brac{N^T}$. Then the linear transformation
 \begin{equation}\label{eq:lin_trans}
    x \to (S^T\, x, Z^T\,x)
  \end{equation}
  sends $x\in\R^n$ to an element $\xi\in\im(N)$ and to an element
  $\eta\in\im(N)^\perp$:
  \begin{equation}
    \label{eq:def_xi_eta_x}
    \xi := S^T\, x \text{ and } \eta := Z^T\, x.
  \end{equation}
  Note that $\xi$ and $\eta$ are unique for given $S$, $Z$ (as
  $\im(N)$ and $\im(N)^\perp$ are complementary subspaces).  Since by
  assumption, $S$ and $Z$ are orthonormal we recover
  \begin{equation}
    \label{eq:x_from_xi_eta}
    x \equiv x\brac{\xi,\eta} = S\, \xi + Z\, \eta.
  \end{equation}
  We further note that by construction
  $\ker(W^T)=\ker(Z^T)=im(N)^\perp$ and hence all elements
  $x\in\omega_{c_0}$ are sent to the same element
  $\eta_0\in\im(N)^\perp$:
  \begin{displaymath}
    x_1,x_2 \in \omega_{c_0} \Rightarrow Z^T\, x_1 = Z^T\, x_2
    =:\eta_0,\; \forall x_1,x_2\in\omega_{c_0}.
  \end{displaymath}
  We now apply the linear transformation (\ref{eq:lin_trans}) to the system (\ref{crm}) to
  obtain: 
  \begin{subequations}
    \begin{align}
      \label{eq:xi_ode_def}
      \dot \xi &= S^T\, \dot x  = S^T N\, v\brac{k, x\brac{\xi,\eta}}
      \\
      \label{eq:eta_ode}
      \dot \eta &= Z^T\, \dot x = Z^T\, N\, v\brac{k,
        x\brac{\xi,\eta}} \equiv 0.
    \end{align}
  \end{subequations}
  That is, $\eta$ is constant, reflecting the invariance of
  $\omega_{c_0}$. We introduce the abbreviation:
  \begin{equation}
    \label{eq:def_g_eta}
    g_{\eta} (k, \xi) := S^T N\, v\brac{k, x\brac{\xi,\eta}},
  \end{equation}
  where in complete analogy to $c_0$ above we interpret $\eta$ as \dred{a} 
  parameter vector. Then every $\eta_0\in\R^{n-r}$
  identifies an $r$-dimensional dynamical system
  \begin{equation}
    \label{eq:xi_ode}
    \dot \xi = g_{\eta_0} (k, \xi).
  \end{equation}
  Now studying the system (\ref{crm}) restricted to a level set
  $\omega_{c_0}$ is equivalent to studying the system
  (\ref{eq:xi_ode}) with $\eta_0 = Z^T\, x_0$ for some $x_0\in\omega_c$.

Solutions $\xi(t,\xi_0)$ of (\ref{eq:xi_ode}) 
give rise to solutions of (\ref{crm})
\begin{displaymath}
  x(t,x_0) = S\, \xi (t,\xi_0) + Z\, \eta_0.\
\end{displaymath}
Since a  solution of the system  (\ref{crm}),  $x(t,x_0) \geq 0$  for
all $t\geq 0$,  it follows  that the corresponding solution $\xi
(t,\xi_0)$  of (\ref{eq:xi_ode}) remains in  the set
\begin{equation}
  \label{eq:def_Om_xi}
  \Omega_{\eta_0} = \left\{ \xi \in \R^r | 
    S\, \xi \geq - Z\, \eta_0
  \right\}.
\end{equation}

The set  $\Omega_{\eta_0}$ has similar properties as the set $\omega_{c_0}$. 
We have the following lemma for $\Omega_{\eta_0}$ which will be used in Corollary~\ref{coro:deg_gxi}. The   proof is
available in \cite{ConradiMincheva2014}.

\begin{lemma}\label{lem:lem2}[$\Omega_{\eta_0}$ convex, compact and forward
  invariant]
  The set $\Omega_{\eta_0}$ is convex, compact and forward
  invariant. 
\end{lemma}

The following lemma compares the number and type of equilibria of
(\ref{crm}) in the set $\omega_{c_0}$ to the number and type of
equilibria of (\ref{eq:xi_ode}) in the set $\Omega_{\eta_0}$.

For a set $A$ we will denote its interior by $\inte(A)$ and its boundary by $\partial A$.
\begin{lemma}\label{lem:lem3}[Equilibrium points.]
  \begin{itemize}
  \item[(a)] 
        A positive point $x^*$ is an equilibrium of (\ref{crm}) in 
      $\omega_{c_0}$ with $c_0 = W^T\, x^*$, if and only if $\xi^*= S^T\, x^*$ is an 
      equilibrium of (\ref{eq:xi_ode}) for $\eta_0 = Z^T\, x^*$ (where
      positivity of $x^*$ entails $\xi^* \in \Omega_{\eta_0}$).
  \item[(b)] The number of equilibria in $\omega_{c_0}$ of  (\ref{crm})
    equals the number of equilibria    of (\ref{eq:xi_ode}) in
    $\Omega_{\eta_0}$.
  \item[(c)] The boundary $\partial \omega_{c_0}$ of the  set
    $\omega_{c_0}$ contains an equilibrium point of  (\ref{crm}) if
    and only if the boundary  $\partial \Omega_{\eta_0}$ of
    $\Omega_{\eta_0}$ contains an equilibrium point  of (\ref{eq:xi_ode}).
  \end{itemize}
\end{lemma}
\begin{proof}
  \begin{itemize}
  \item[(a)] 
    This follows since $S^T N v(k, x^*) =g_{\eta}(k, \xi^* )=0$,
    where $x^* \in \omega_{c_0}$ corresponds to $\xi^* \in
    \Omega_{\eta_0}$ such that  $\eta_0 =Z^T x^*$. 
  \item[(b)] This follows from the fact that $S$ and $Z$ are full rank
    matrices, and thus the correspondence between $x$ and $(\xi,\eta)$
    is one-to-one.
  \item[(c)] 
    We prove the contrapositive.
    An equilibrium  $x^* >0$ of  (\ref{crm}) is not on $\partial
    \omega_{c_0}$ if and only if  an equilibrium $\xi^*$ of
    (\ref{eq:xi_ode}) is such that  $S \xi^* > -Z \eta$.
  \end{itemize}
\end{proof}

  From hereon we will assume level sets $\omega_{c_0}$ do not contain
  equilibria with zero coordinates, referred to as {\it boundary equilibria}.
  \begin{assumption}\label{ass3}
    Here we study only DE models of biochemical mechanisms such that
    the level set $\omega_{c_0}$ does not contain any boundary
    equilibria, that is,
    \begin{displaymath}
   \mbox{  if }   x \in \partial \omega_{c_0} \Rightarrow Nv(k,x) \neq 0.
    \end{displaymath}
\end{assumption}

  The following corollary follows by   Lemma \ref{lem:lem3}  and
  Assumption~\ref{ass3}.
  \noindent
  \begin{corollary}\label{cor_num_eq}
        Under Assumption 3 the boundary $\partial \Omega_{\eta_0}$ of $\Omega_{\eta_0}$ does
    not contain an equilibrium of (\ref{eq:xi_ode}).  \\
    The number of equilibria of (\ref{crm})   in $\inte (\omega_{c_0})$  equals the 
    number of equilibria of (\ref{eq:xi_ode}) in $\inte (\Omega_{\eta_0})$.
  \end{corollary}
  
 \begin{remark}
    In essence, in Sec.~\ref{sec:degree_g} we will study the number of
    equilibria of the reduced system (\ref{eq:xi_ode}) in $\inte
    (\Omega_{\eta_0})$. By Corollary~\ref{cor_num_eq}  we will obtain
    the corresponding result on the number of equilibria of  the system 
    (\ref{crm})   in $\inte (\omega_{c_0})$.
  \end{remark}

\section{The Jacobian $J$ parametrized at $(v,x)$ and its projection
  on $(\xi, \eta) $ space}
\label{sec:jac_ma}

The Jacobian matrix  $J (k,x)$ of (\ref{crm}) has entries 
\begin{equation}\label{jac1}
J_{il} (k,x) = \sum_{j=1}^{m} N_{ij} \alpha_{lj} k_j x_{1}^{\alpha_{1j}} \ldots x_{l}^{\alpha_{lj}-1} \ldots x_{n}^{\alpha_{nj}}. 
\end{equation}
Recall (\ref{rf}), then the Jacobian  can be written also as  
\begin{equation}\label{jac}
J_{il}(k,x) = J_{il} (v,x)= \sum_{j=1}^{m} N_{ij} \alpha_{lj} \frac{v_j}{x_l} .
\end{equation} 
Note that if the concentrations  $x$  and the rate functions $v
(k,x)$ are evaluated at a positive equilibrium, they are positive and can be 
used as parameters in (\ref{jac}).

For example, the Jacobian matrix of the right-hand side of the system  (\ref{eq:de_cyc}) is  
\begin{equation}\label{jac-ex}
J (v,x) = \left(
\begin{matrix}
-\frac{ v_3}{x_1} -\frac{v_5}{x_1}&\frac{2v_1}{x_2} +\frac{v_4}{x_2}  & 2\frac{v_1}{x_3}+ \frac{v_2}{x_3} \\
\frac{v_5}{x_1}  & -\frac{v_1}{x_2} -\frac{v_4}{x_2}  & -\frac{v_1}{x_3}  \\
\frac{v_3}{x_1} &-\frac{v_1}{x_2} &    -\frac{v_1}{x_3}  -\frac{v_2}{x_3} 
\end{matrix}
\right).
\end{equation}

For the remainder of this contribution we make the following assumption.
  \begin{assumption}\label{ass4}
    We assume that $\im(J(v,x)) = \im(N)$.
  \end{assumption}

Since the rank of  $N$ is  $r$, it follows 
under the above assumption 
that $rank (J(v,x)) \leq r$ and that the characteristic polynomial of 
the Jacobian (\ref{jac}) is
\begin{equation}
  \label{eq:cp}
  \det(\lambda I -J(v,x))=  \lambda^{n-r}\, \left( \lambda^r+ a_1\, \lambda^{r-1} + \ldots + a_{r-1}
    \lambda + a_r \right) = \lambda^{n-r} q(\lambda),
\end{equation}
where the coefficients $a_i = a_i (v,x)$, $i=1$,\ldots, $r$ are
computed as the sum of all principal minors of order $i$ of the
negative Jacobian $-J(v,x)$ \cite{gantmacher1960}.

The coefficients   $a_i(v,x)$ of (\ref{eq:cp})  are rational functions
in $x$ and $v$ by (\ref{jac}). For example, the non-zero  coefficients
of the characteristic polynomial of the Jacobian (\ref{jac-ex}) are
\begin{align}
  a_1(v,x)& = \frac{ v_3 +v_5}{x_1} + \frac{v_1+v_4}{x_2}
  +\frac{v_1+v_2}{x_3} \label{eq:ex_a1} \\ 
  a_2 (v,x)&=   \frac{v_1v_3 - v_1v_5 +v_3v_4   }{x_1 x_2} +
  \frac{-v_1 v_3 +v_1v_5 +v_2v_5}{x_1 x_3}  +  \frac{v_1 v_2 +v_1v_4
    +v_2v_4}{x_2 x_3}. \label{eq:ex2}
\end{align}

  It is easy to verify, that the Jacobian with respect to $\xi$ of
  $g_\eta(k,\xi)$ from (\ref{eq:def_g_eta}) is given by
  \begin{equation}
    \label{eq:def_G_eta}
    G_\eta(v,\xi) = S^T\, J(v,x)\, S,
  \end{equation}
  where we have suppressed the $\xi$, $\eta$ dependence of $x$.

The relation  between the coefficient $a_r(v,x) $ and the determinant
of the negative Jacobian $-G_{\eta}(v,\xi)$ of the right-hand side of
the  reduced system $ \dot{\xi}=g(v,\xi)$  is considered in the next
lemma. A special case of this lemma with $r=6$ is available in
\cite{ConradiMincheva2014}. 

\begin{lemma}
  \label{lem:det_G_a6}
  The following equivalence is true
  \begin{displaymath}
    \det\brac{-G_\eta(v,\xi)} \equiv \det\brac{-S^T\, J(v,x)\, S} \equiv
    a_r(v,x)\ .
  \end{displaymath}
\end{lemma}
\begin{proof} \mbox{}\\
  Let $\lambda_1 (v,x)$, \ldots, $\lambda_n (v,x)$ be the eigenvalues of $J(v,x)$
  and hence the roots of the characteristic polynomial (\ref{eq:cp}), 
  where we assume that $\lambda_i\equiv\lambda_i(v,x)$. By (\ref{eq:cp}), 
   there exist $(n-r)$ trivial  eigenvalues that are
  identically  zero for all values of \brac{v,x}  and $r$ non-trivial  eigenvalues that are nonzero for
  some values of \brac{v,x}.  For simplicity we assume that $\lambda_i$, $i=1$,
  \ldots, $r$ are nontrivial and $\lambda_i$, $i=n-r+1,\ldots,  n$ are
  trivial.
  
  First we show that $a_r$ is the product of the nontrivial
  eigenvalues $$a_r=\prod_{i=1}^r \lambda_i.$$ 
  If we  write the polynomial $q(\lambda)$
  from the  characteristic polynomial (\ref{eq:cp}) in factored
  form and substitute  $\lambda =0$ we obtain  $a_r=\prod_{i=1}^r
  \lambda_i.$
  
  Next we show that $\det(-S^T J (v,x) S ) = \prod_{i=1}^r \lambda_i$,
  which will prove the claim $\det(-S^T J (v,x) S ) = a_r$. For this
  purpose we apply the orthonormal transformation
  $\phi=(S, Z)$ to $J(v,x)$. 
  The transformed matrix has a block form
  \begin{displaymath}  
    \phi^T J \phi=\left[ 
      \begin{array}{c|c} 
        S^T J S& S^TJ Z \\ \hline 
        0 & 0
      \end{array}\right].
  \end{displaymath}
  Since $\phi$ is orthonormal,  $J(v,x)$ and $\phi^T J(v,x) \phi$ have
  the same eigenvalues.  Both matrices $J(v,x)$ and $\phi^T J(v,x) \phi$
  have $(n-r)$ trivial eigenvalues. Thus  $$\det(\lambda I - \phi^T J (v,x)\phi 
  ) =\lambda^{n-r} \det( \lambda I -  S^T J (v,x) S  ) =\lambda^{n-r}
  \tilde{q}(\lambda).$$  If  $ \tilde{q}(\lambda)$ is written in
  factored form and we let $\lambda =0$ we obtain $\det(-S^T J(v,x) S)
  = \prod_{i=1}^r \lambda_i$. Thus  $\det(-S^T J (v,x) S ) = a_r$.
\end{proof}

\section{The degree of $g_{\eta}(k,\xi)$}
\label{sec:degree_g}

Let $U  \subset \mathbb{R}^n$ be an open and  bounded  set. The
closure of $U$ will be denoted by $\bar{U}$ and the
boundary of $U$ by  $ \partial U$. Thus, $\bar{U}
= U \cup \partial U$ is a compact set.

Let  $F:\bar{{U}} \to \mathbb{R}^n$  be a smooth function, where using  the
usual notation we write $F(x) \in C^1 (\bar{{U}})$. We denote the
Jacobian matrix of $F(x)$ by 
\begin{equation}\label{eq:jac}
\tilde{J}(x) = \left[\frac{\partial F_i}{\partial  x_j}\right]
\end{equation} 
  and its determinant by $\det (\tilde{J}(x))$. A point $x \in {U}$
is a {\it regular point} for $F(x)$ if $\det (\tilde{J}(x)) \ne 0$. A point $y
\in \mathbb{R}^n$ is called a {\it regular value} if all $x \in
U$ such that  $F(x)=y$ are regular. 

Next we define the  {\it (topological or Brouwer) degree of $F(x)$} \cite{Deimling2010}, denoted by $\deg(F)$. 
In the next definition we   use the sign function $\sign: \R \to \{ -1,0,1\}$. 

\begin{definition}\label{def:topo-degree}[(topological) degree]
  If $y \notin F(\partial {U})$ and $y$ is a regular value, the
  {\it   degree} of F is defined by 
  \begin{equation}\label{def:deg} 
    \deg (F) =\deg (F, U, y) =\sum_{F(x)=y} \sign (\det(-\tilde{J}(x))). 
  \end{equation}
 \end{definition}

   \begin{remark} Note that we use  $\sign (\det (-\tilde{J} (x)))$  in (\ref{def:deg}) in place  of $\sign (\det (\tilde{J} (x))) $
to avoid the  case of the degree depending on the dimension $n$ of $\mathbb{R}^n$ similarly to \cite{hofbauer1990}. \end{remark}

  The sum in
  (\ref{def:deg}) is over all solutions $x \in {U}$  of $F(x) =y$
  such that $\det (-\tilde{J}(x)) \ne 0$. If $F(x) =y$ does not have solutions
  $x \in {U}$, then we set $\deg(F) =0$. 
  Since we are interested in the equilibrium solutions  $x^*$ of $\dot{x} =F(x) $ that satisfy
$F(x^*)=0$, we will  let $y=0$ in (\ref{def:deg}).

Next we study the degree of the function $ g_{\eta}\brac{v,\xi}$
defined in (\ref{eq:xi_ode}).  Recall that 
  $x\equiv x(\eta,\xi)$ and that 
$a_r (v,x) =\det (-G_{\eta}
(v, \xi))$ by Lemma~\ref{lem:det_G_a6}, where $G_{\eta} (v,\xi) $ is
the Jacobian of the function $ g_{\eta}\brac{v,\xi}$ 
  given in (\ref{eq:def_G_eta}). 

The next lemma is similar to  Lemma 5.4 available in the Supporting
Information of \cite{ConradiMincheva2014}. 
\begin{remark}  Note that in the lemma and corollaries below  the assumption that $\partial \Omega_{\eta}$ does not contain any boundary equilibrium of the system (\ref{eq:xi_ode}) is 
automatically satisfied by Assumption~\ref{ass3} and Corollary~\ref{cor_num_eq}.  
 \end{remark}

\begin{lemma}\label{lem:lem5}
    Let $g_\eta$ be as in (\ref{eq:def_g_eta}).
    Fix  $\eta\in\R^{n-r}$ 
    and assume that the boundary $\partial \Omega_{\eta}$ does not
     contain any equilibria of (\ref{eq:xi_ode}).
    If all equilibria $\xi\in\Omega_{\eta}$ are regular, then
  \begin{equation}
    \label{eq:def_deg_g_eta}
    \deg(g_\eta ,\inte(\Omega_\eta),0) = \sum_{ \{ \xi \in
      \inte(\Omega_\eta) | g_{\eta} (\xi) =0\}} \sign (a_r(v,x(\eta,\xi))).
  \end{equation}
\end{lemma}
\begin{remark}
By Corollary~\ref{cor_num_eq} the equilibria of (\ref{eq:xi_ode}) are in $\inte (\Omega_{\eta})$.
Therefore the degree of $g_{\eta} (\xi)$, $\deg (g_\eta,\inte(\Omega_\eta), 0)$  given by (\ref{eq:def_deg_g_eta}) is well defined. 
\end{remark}
We  obtain the following corollaries on the degree of $g_{\eta}
(k,\xi)$ where $\xi \in\Omega_{\eta}$. Similar corollaries  for the
special case of $r=6$ are available in \cite{ConradiMincheva2014}.

First we need the following theorem on the homotopy invariance of the degree \cite{CraciunHelton2008}.
\begin{theorem}\label{thm:deg_inv}
Let $U \subset \mathbb{R}^n$ be a bounded and open set. Let $H(x,s): \bar{U} \times [0,1] \to \mathbb{R}^n $, 
be a continuously varying set of functions such that $H(x,s)$ does not have any zeroes on the boundary of $U$ for all $s \in [0,1]$. Then $\deg(H(x,s))$
is constant for all $s \in [0,1]$.
\end{theorem}
\begin{corollary}
  \label{coro:deg_gxi}
    Let $g_\eta$ be as in (\ref{eq:def_g_eta}) 
     and assume that the boundary $\partial \Omega_{\eta}$ does not
     contain any equilibria of (\ref{eq:xi_ode}).Then the 
    the following holds true:
  \begin{equation}
    \label{eq:deg_g}
    \deg\brac{g_{\eta},\inte(\Omega_{\eta}),0} = 1.
  \end{equation}
\end{corollary}
\begin{proof}
 Since $f(k,x)$ is smooth on $\omega_{c_0}$, therefore  $g_{\eta}(k,\xi)$ is
  smooth on  $\Omega_{\eta}$. Let $k$ be fixed but arbitrary so that $g_{\eta}(k,\xi)=g_{\eta} (\xi)$.\\
  We have   $\Omega_{\eta} =\inte(\Omega_{\eta}) \cup \partial
  \Omega_{\eta} $, where $\inte(\Omega_{\eta})$ is the interior of
  $\Omega_{\eta}$ and $\partial \Omega_{\eta}$ is the boundary of
  $\Omega_{\eta}$.  By  Lemma~\ref{lem:lem2},   $\inte(\Omega_{\eta})$ is bounded.
We follow the proof  of \cite[Lemma 2]{DeSontag2007}. 
Let $\bar{\xi} \in \inte(\Omega_{\eta})$ be an arbitrary point and consider the  function
$$ G(\xi) = \bar{\xi} -\xi.$$
By Definition~\ref{def:topo-degree} it follows that 
\begin{equation}\label{eq:deg_G}
\deg(G, \inte(\Omega_\eta), 0)=1.
\end{equation} 
Next we show that $g_{\eta}$ and $G$ are homotopic. 
We define the following homotopy  $$H(\xi,s) =  s g_{\eta}(\xi) + (1-s) G(\xi)  $$ where $0 \leq s \leq 1$.
Therefore $H(\xi,s)$ is continuous on  $\Omega_{\eta} \times [0,1]$, $H(\xi,0)= G(\xi)$ and $H(\xi,1) = g_{\eta} (\xi)$.
To apply Theorem~\ref{thm:deg_inv}
we  need to show that $H(\xi,s) \ne 0$ for all $\xi \in \partial \Omega_{\eta}$ and for all $s \in [0,1]$. 
The latter is true if $s=0$ since $\bar{\xi} \in \inte(\Omega_\eta)$ and if  $s=1$ by Corollary~\ref{cor_num_eq}.
Suppose it is not true if $s \in (0,1)$, then there exists $\tilde{\xi} \in \partial \Omega$ and $\tilde{s} \in (0,1)$ such that 
$$ g_{\eta} (\tilde{\xi}) = -\frac{1-s}{s} G(\tilde{\xi}). $$
 By the convexity of $\Omega_{\eta}$ it follows that $G(\xi)$ points strictly 
inwards at $\xi=\tilde{\xi}$. Thus  $g_{\eta} (\xi)$ at $\xi =\tilde{\xi}$ points strictly outwards. This is a contradiction since 
$\Omega_{\eta}$ is forward invariant by Lemma~\ref{lem:lem2}. 
Thus the claim  in (\ref{eq:deg_g}) follows  by equation (\ref{eq:deg_G}) since the degree is
homotopy invariant by Theorem~\ref{thm:deg_inv}. 
\end{proof}

\begin{corollary}\label{cor:cor2}
 Let $\eta$ and $k$ be given 
    and note that $v\equiv v(k,x)$. Assume
    that the boundary $\partial \Omega_{\eta}$ does not contain any
    equilibria of (\ref{eq:xi_ode}).
  If $a_r(v,x(\eta,\xi))>0$ for all $\xi \in\inte(\Omega_{\eta})$, then
  the equation
  \begin{displaymath}
    g_{\eta}(v,\xi) = 0, \xi\in\Omega_{\eta}
  \end{displaymath}
  has a unique solution.
  
  If all 
  solutions of $g_{\eta}(v,\xi)=0$, $\xi\in \inte(\Omega_{\eta})$
  are regular, then the number of
  solutions
  in $ \inte(\Omega_{\eta})$ is odd.
\end{corollary}

\begin{proof}
  Recall that by Lemma~\ref{lem:det_G_a6}, $\det (-G_{\eta}) = a_r (v,
  x(\eta,\xi))$. Suppose that $a_r (v, x(\eta,\xi))>0 $ for all $\xi \in
  \inte(\Omega_{\eta})$. Then by  Corollary~\ref{coro:deg_gxi},
  $\deg(g_\eta,  \inte(\Omega_{\eta}), 0) =1$. Thus  by
  (\ref{eq:def_deg_g_eta}), it follows that the number of 
  equilibria in $\inte(\Omega_{\eta})$ equals one.

  If all equilibria $\xi^*$ are regular, then $a_r (v, x(\eta,\xi))$ is
  either positive or negative at $\xi=\xi^*$.
  Since  $\deg(g_\eta,  \inte( \Omega_{\eta}) ,0) =1$ by
  Corollary~\ref{coro:deg_gxi},  it follows by
  (\ref{eq:def_deg_g_eta}) that the number of equilibria has to be
  odd.
\end{proof}

  Now we turn to the system (\ref{crm}). In the next theorem we show that the system (\ref{crm}) has an  interior equilibrium solution in any set $\omega_{c_0}$
  where $c_0$ is fixed.
  \begin{theorem}
    Let $f(k,x)$ be as in (\ref{crm}) and 
    recall Assumption~\ref{ass3}. If
    \begin{displaymath}
      a_r(v,x) >0\text{, for all } v>0 \text{ and } x>0,
    \end{displaymath}
    then
    \begin{displaymath}
      f(k,x) = 0, x\in \omega_{c_0}
    \end{displaymath}
    has a unique solution for all $c_0>0$. 
  \end{theorem}
  \begin{proof}
    Pick $c_0>0$ and choose any $x_0\in\omega_{c_0}$. Compute $\xi_0 =
    S^T\, x_0$, $\eta_0=Z^T\, x_0$. By assumption
    \begin{displaymath}
      a_r(v,x(\eta_0,\xi)) >0, \forall v>0 \text{ and } \forall
      \xi\in\Omega_{\eta_0}. 
    \end{displaymath}
    Thus, by Corollary~\ref{cor:cor2}, $g_{\eta_0}(k,\xi)$,
    $\xi\in\Omega_{\eta_0}$ has a unique solution $\xi^*$.  And by 
    Lemma~\ref{lem:lem3}  (a)-(b),  an equilibrium  $x^*=S\, \xi^*+Z\, \eta_0$ of 
    $\dot{x}= f(k,x^*)$, $x^*\in \inte (\omega_{c_0})$ is unique.
      \end{proof}

\section{The bipartite digraph of a biochemical mechanism}\label{sec:bip_gr}

For the convenience of the reader,  in
this section we present  definitions regarding the bipartite digraph of a biochemical mechanism (\ref{cr})   \cite{VolpertIvanova,Mincheva2007}. 
To illustrate the definitions  we will continue to use as an  example the mechanism (\ref{eq:br_cyc}).

A  {\it directed  bipartite  graph}  (bipartite digraph) has a node set that consists of two disjoint  subsets,
$V_1$ and $V_2$, and each of its directed edges (arcs) has one end in 
$V_1$ and the other in $V_2$ \cite{fh69}.

The {\it bipartite digraph} $G$ of a biochemical reaction network (\ref{cr})
is defined as follows.
The nodes are separated  into two sets, one for the chemical
species $V_1 =\{ A_1, A_2,\ldots, A_n \}$ and one for the elementary  reactions
$V_2 = \{ B_1, B_2,\ldots, B_m\}$. 
We draw an arc   from $A_k$ to $B_j$ if and only if species $A_k$ is a
reactant in reaction $j$, i.e., \ if  the stoichiometric coefficient  $\alpha_{kj}>0$ in (\ref{cr}).
Similarly, we draw an arc  from $B_j$ to $A_i$ if and only if $A_i$ is a
product in reaction $j$, i.e., \ if  the stoichiometric coefficient $\beta_{ij}>0$ in (\ref{cr}). 
Therefore the  set of arcs $E(G)$ consists of arcs such as  $(A_k,
B_j)$ and $(B_j, A_i)$. 
Hence the  bipartite digraph can be defined as $G= \{ V, E(G)\}$ where $V=V_1 \cup V_2$ is the set of nodes and  $E(G)$ is the set of arcs. If an arc is not weighted explicitly, we assume that its weight equals $1$.
The  corresponding  bipartite digraph of the   mechanism (\ref{eq:br_cyc})  is  shown in Figure \ref{fig:br_cyc}.

 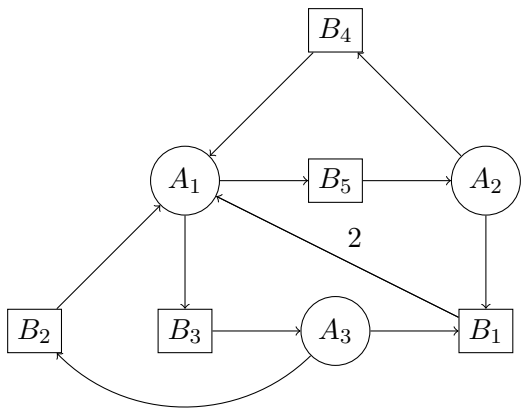
\begin{figure}
\centering
 \begin{tikzpicture}[scale=1,auto]
    \node  [shape=circle,draw]  (n1) at (10,12)  {$A_1$};
     \node  [shape=rectangle,draw]  (b1) at (12,12)  {$B_5$};
     \node [shape=rectangle,draw]   (b6) at  (12,14) {$B_4$};
   \node  [shape=circle,draw]  (n2) at (14,12)  {$A_2$};
      \node  [shape=rectangle,draw]  (b3) at (14,10)  {$B_1$};
   \node   [shape=circle,draw]  (n4) at (12,10)  {$A_3$};
      \node  [shape=rectangle,draw]  (b4) at (10,10)  {$B_{3}$};
     \node  [shape=rectangle,draw]  (b7) at (8,10)  {$B_{2}$};
       \draw [->] (n1)  to    node {}   (b1);
                   \draw [->] (b1)  to    node {}   (n2);
                    \draw [->] (n2)  to    node {}   (b3);
                    \draw [->] (b3)  to    node {}   (n1);
                     \draw [] (n1)  to    node {2}   (b3);
                    \draw [->] (n4)  to    node {}   (b3);
                    \draw [->] (n1)  to    node {}   (b4);
                    \draw [->] (n2)  to    node {}   (b6);
                     \draw [->] (b6)  to    node {}   (n1);
                    \draw [->] (b4)  to    node {}   (n4);
                     \draw [->] (n4)  to [bend left=45]   node {}   (b7);
                           \draw [->] (b7) to node {}    (n1); 
 \end{tikzpicture}
\caption{Bipartite graph of the reaction mechanism (\ref{eq:br_cyc}).}
\label{fig:br_cyc}
\end{figure}

The element $[A_k, B_j]$ is an {\it edge}  if    
$\alpha_{kj}  > 0$, i.e., if species $A_k$ is a reactant in reaction $j$. The {\it weight of an  edge} $E= [A_k,B_j]$ is defined as 
\begin{equation}\label{ke}
K_E = -\alpha^{2}_{kj}.
\end{equation}
For example, the edge $E=[A_1,B_5]$ corresponding to the arc $(A_1,B_5)$ in Figure~\ref{fig:br_cyc} has weight $K_E =-1$. 

 If $\alpha_{kj} \beta_{ij} > 0$,  then
the arcs $(A_k, B_j)$ and $(B_j, A_i)$ form a {\it positive path}
 $[A_k ,B_j ,A_i]$ that  corresponds to the production of
$A_i$ from $A_k$ in a reaction $j$. 
The weight of the positive path $[A_k ,B_j ,A_i]$ is defined as $\alpha_{kj} \beta_{ij}$.
For example, the positive path $[A_1,B_5,A_2]$ in  Figure \ref{fig:br_cyc} has weight $1$. 

If $\alpha_{kj}\alpha_{ij} > 0$, then 
the  arcs $(A_k, B_j)$ and $(A_i, B_j)$ form a
{\it negative path} $[\overline{A_k, B_j, A_i}]$ that 
corresponds to $A_k$ and $A_i$ interacting as reactants in reaction $j$.
The weight of the negative path $[\overline{A_k, B_j, A_i}]$ is defined as  $-\alpha_{kj} \alpha_{ij}$.
Note that the negative paths $[\overline{A_k, B_j, A_i}]$ and $[\overline{A_i, B_j,A_k}]$ are considered to be different since they start at a different
species node.  For example, both $[\overline{A_2,B_1,A_3}]$ and  $[\overline{A_3,B_1,A_2}]$ in Figure \ref{fig:br_cyc} are negative paths with weight $-1$. We note that the direction of the arcs is followed in the positive paths but not in the  negative paths.

A {\it cycle} $C$ of $G$ is a 
sequence of distinct paths with the last  species 
node of each path being the same as the first species node  of the next
path $ C = \{ (A_{i_1},B_{j_1}, A_{i_2})$, $(A_{i_2}, B_{j_2}, A_{i_3})$,$\ldots$,
$(A_{i_{k-1}}, B_{j_{k-1}}, A_{i_k})$, $(A_{i_k}, B_{j_k}, A_{i_1}) \}$.
A cycle will be denoted by $C=\binom{A_{i_1},
A_{i_2},\ldots,A_{i_k}}{B_{j_1}, B_{j_2},\ldots, B_{j_k}}$, 
where the number of species nodes defines its {\it order}.
The set of species nodes in a cycle   is distinct, but there may be a repetition among the 
reaction nodes. This is because negative paths containing the same nodes are considered 
different depending on the starting species node.  
For example, $C=\binom{A_{2},
A_{3}}{B_{1}, B_{1}}$ in Figure \ref{fig:br_cyc} is a cycle formed by the two negative paths $[\overline{A_2, B_1, A_3}]$ and $[\overline{A_3, B_1, A_2}]$.  

\begin{figure}
 \centering
 \begin{tikzpicture}[scale=1,auto]
    \node  [shape=circle,draw]  (n1) at (10,12)  {$A_1$};
     \node  [shape=rectangle,draw]  (b1) at (12,12)  {$B_5$};
   \node  [shape=circle,draw]  (n2) at (14,12)  {$A_2$};
      \node  [shape=rectangle,draw]  (b3) at (14,10)  {$B_1$};
       \draw [->] (n1)  to    node {}   (b1);
                           \draw [->] (b1)  to    node {}   (n2);
                    \draw [->] (n2)  to    node {}   (b3);
                    \draw [->] (b3)  to    node {}   (n1);
                     \draw [] (n1)  to    node {2}   (b3);
 \end{tikzpicture}
\caption{Cycle   $C \binom{A_{1}, A_{2}}{B_{5}, B_{1}}$  of the bipartite graph of the reaction mechanism (\ref{eq:br_cyc}).}
\label{fig:cyc_br_cyc}
\end{figure}
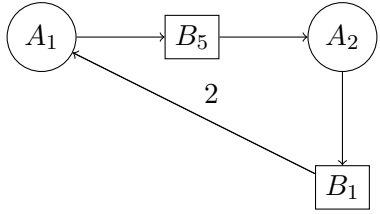

A cycle is {\it  positive} if it contains an even number  of
negative paths and {\it negative} if it contains an odd number of
negative paths. The  sign of a   cycle $C$ can also be determined by the {\it cycle weight}
which is a product of all corresponding weights of  negative and positive paths  of $C$
\begin{equation}\label{kc}
K_C = \prod_{\overline{[A_k, B_j, A_i]} \in C}( - \alpha_{kj} \alpha_{ij}) \prod_{[A_k, B_j , A_i] \in C} \alpha_{kj} \beta_{ij}.
\end{equation}
For example,   $C=\binom{A_{1}, A_{2}}{B_{5}, B_{1}}$ in   Figure \ref{fig:cyc_br_cyc}  is a  positive cycle of order $2$ with weight $K_C=2$.

A {\it subgraph}  $g =\{\text{L}_1 , \text{L}_2 ,\ldots, \text{L}_s  \}$ of $G$
consists of edges or cycles $\text{L}_i$, $i=1,\ldots , s$, 
where each species  is the beginning of 
only one edge,  or one  path participating in a cycle.
The number of species nodes   in a subgraph is defined as  its {\it  order}.
The {\it subgraph weight} is defined as 
\begin{equation}\label{kg}
K_g = (-1)^c \prod_{C \in g} K_C \prod_{E\in g} (-K_E),
\end{equation}
 where $c$ is the number of cycles in $g$, $K_C$ is  the cycle weight (\ref{kc}) and $K_E$ is the edges weights (\ref{ke}) of the cycles and edges in $g$. 
 For example, the subgraph $g =\{ [A_3, B_2],  \binom{A_{1}, A_{2}}{B_{1}, B_{5}} \}$ 
 with weight $K_g=-2$ is shown in  Figure \ref{fig:sub_br_cyc} .

Since more than one path can exist between species nodes via different reaction nodes 
in a bipartite digraph, the number of subgraphs through the same node sets may be greater than one.
The set of all subgraphs $g$ of order $k$ with the
same species nodes $\bar{V}_1 = \{ A_{i_1}, \ldots A_{i_k} \}$ and  reaction nodes 
$\bar{V}_2 = \{ B_{j_1}, \ldots B_{j_k} \}$  sets  is called a 
{\it fragment} of order $k$ and is
denoted by $S_k\binom{i_1,\ldots,i_k}{j_1,\ldots,j_k}$. 
For a fragment $S_k \binom{i_1,\ldots,i_k}{j_1,\ldots,j_k}$
we define the number
 \begin{equation}\label{ksk}
 K_{S_k}=\sum_{g \in S_k}  K_{g}
 \end{equation}
 as the {\it fragment weight}. 
If $K_{S_k} <0$, then $S_k$ is defined as a {\it  critical fragment}.

For example, the fragment $S_2\binom{1, 2}{5,1}$ shown in
Figure~\ref{fig:cf_br_cyc}  together with its two
subgraphs 
$g_1 =C_2=\binom{A_1,  A_2}{B_5, B_1}$ and
$g_2=\{[A_1, B_5], [A_2, B_1]\}$.
The  first  subgraph $g_1$ is a positive cycle and thus it has a negative weight.  Therefore 
$S_2 \binom{1, 2}{5,1}$ is a critical fragment since 
$$K_{S_2}= \sum_{g \in S_2} K_g=K_{g_1}+K_{g_2}= -2+1=-1<0.$$  

 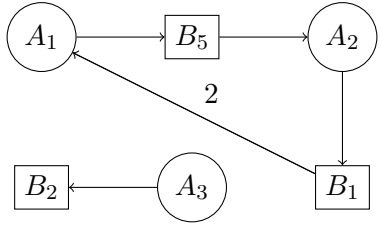
\begin{figure}
\centering
 \begin{tikzpicture}[scale=1,auto]
    \node  [shape=circle,draw]  (n1) at (10,12)  {$A_1$};
     \node  [shape=rectangle,draw]  (b1) at (12,12)  {$B_5$};
   \node  [shape=circle,draw]  (n2) at (14,12)  {$A_2$};
      \node  [shape=rectangle,draw]  (b3) at (14,10)  {$B_1$};
   \node   [shape=circle,draw]  (n4) at (12,10)  {$A_3$};
      \node  [shape=rectangle,draw]  (b4) at (10,10)  {$B_{2}$};
       \draw [->] (n1)  to    node {}   (b1);
                   \draw [->] (b1)  to    node {}   (n2);
                    \draw [->] (n2)  to    node {}   (b3);
                    \draw [->] (b3)  to    node {}   (n1);
                     \draw [] (n1)  to    node {2}   (b3);
                            \draw [->] (n4)  to    node {}   (b4);
 \end{tikzpicture}
\caption{Subgraph   $g =\{ [A_3, B_2],  \binom{A_{1}, A_{2}}{B_{5}, B_{1}} \}$  of the bipartite graph of the reaction mechanism (\ref{eq:br_cyc}).}
\label{fig:sub_br_cyc}
\end{figure}

In \cite{VolpertIvanova,Mincheva2007} it is shown that the coefficients
of the characteristic polynomial  (\ref{eq:cp}) have the following graph-theoretic representation.
\begin{theorem}\label{thm1}
A coefficient of the characteristic polynomial of the Jacobian (\ref{jac}) can be written as  
\begin{equation} \label{fs}
a_k (v,x) = \sum_{S_k \binom{i_1,\ldots,i_k}{j_1,\ldots,j_k}} K_{S_k} 
\frac{v_{j_1} \ldots v_{j_k} }{x_{i_1}\ldots x_{i_k}} , \quad k=1, \ldots , n .
\end{equation} 
where $S_k \binom{i_1,\ldots,i_k}{j_1,\ldots,j_k}$ is a fragment of order $k$ and $K_{S_k}$ is the fragment's  weight.
\end{theorem}
Note that similar terms in $a_k(v,x)$ have been combined using summation over the  subgraphs of a fragment (\ref{ksk}),    and (\ref{fs}) 
is in a simplified form.  
It follows by (\ref{fs}) that  the correspondence between a fragment $S_k\binom{i_1,\ldots,i_k}{j_1,\ldots,j_k}$ and a non-zero term in $a_k(v,x)$  is one-to-one.  For example, the first negative term in the coefficient  (\ref{eq:ex2})
corresponds to the critical fragment $S_2 \binom{1,2}{5,1}$ shown in Figure \ref{fig:cf_br_cyc}. 

The next corollary  follows immediately by Theorem \ref{thm1}.

\begin{corollary}\label{cor:cor3}
   Recall the function $f$ from (\ref{crm}) with Jacobian \dred{$J$} as
    in (\ref{jac}).
  The last (not identically zero)
  coefficient of the characteristic polynomial (\ref{eq:cp})
   can be written as
    \begin{equation}
      \label{gtrep-coef11}
      a_{r} (v , x) =  \sum_{{S_r \binom{i_1, i_2, \ldots ,i_r }{j_1,
            j_2,\ldots j_r}} \in G} K_{S_r } \frac{v_{j_1} \ldots
        v_{j_r}}{x_{i_1} \ldots x_{i_r} },
    \end{equation}
  where $r$ is the rank of the stoichiometric matrix $N$. 
\end{corollary}

  \begin{remark}
    Recall that $x\equiv x(\nu,\xi)$. Hence, for fixed values of
    $\eta$, one may consider $a_r(v,x)$ as a function of $\xi$. 
    We further note that $a_r(v,x)$ depends on $k$, as $v\equiv
    v(k,x)$. 
  \end{remark}

 \begin{figure}
\centering
 \begin{tikzpicture}[scale=1,auto] 
    \node  [shape=circle,draw]  (n1) at (6,12)  {$A_1$};
     \node  [shape=rectangle,draw]  (b1) at (8,12)  {$B_5$};
   \node  [shape=circle,draw]  (n2) at (10,12)  {$A_2$};
      \node  [shape=rectangle,draw]  (b3) at (10,10)  {$B_1$};
        \draw [->] (n1)  to    node {}   (b1);
                            \draw [->] (b1)  to    node {}   (n2);
                    \draw [->] (n2)  to    node {}   (b3);
                    \draw [->] (b3)  to    node {}   (n1);
                     \draw [] (n1)  to    node {2}   (b3);
                      \node  [shape=circle,draw]  (n1) at (12,12)  {$A_1$};
     \node  [shape=rectangle,draw]  (b1) at (14,12)  {$B_5$};
   \node  [shape=circle,draw]  (n2) at (16,12)  {$A_2$};
      \node  [shape=rectangle,draw]  (b3) at (16,10)  {$B_1$};
        \draw [->] (n1)  to    node {}   (b1);
                            \draw [->] (b1)  to    node {}   (n2);
                    \draw [->] (n2)  to    node {}   (b3);
                    \draw [->] (b3)  to    node {}   (n1);
                     \draw [] (n1)  to    node {2}   (b3);
                      \node  [shape=circle,draw]  (n1) at (18,12)  {$A_1$};
     \node  [shape=rectangle,draw]  (b1) at (20,12)  {$B_5$};
      \draw [->] (n1)  to    node {}   (b1);
 \node  [shape=circle,draw]  (n2) at (18,10)  {$A_2$};
      \node  [shape=rectangle,draw]  (b3) at (20,10)  {$B_1$};
       \draw [->] (n2)  to    node {}   (b3);
            \node [ ] at (7,10) {(a)};
            \node [ ] at (13,10) {(b)};
           \node [ ] at (17,10) {(c)};
 \end{tikzpicture}
\caption{The critical fragment $S_2 \binom{1,2}{5,1}$  (shown in (a)) together with its two subgraphs  $g_1=C_2= \binom{A_{1}, A_{2}}{B_{5}, B_{1}} \}$ (shown in (b))and $g_2 =\{ [A_1,B_5], [A_2,B_1]\}$ (shown in (c)).}
\label{fig:cf_br_cyc}
\end{figure}
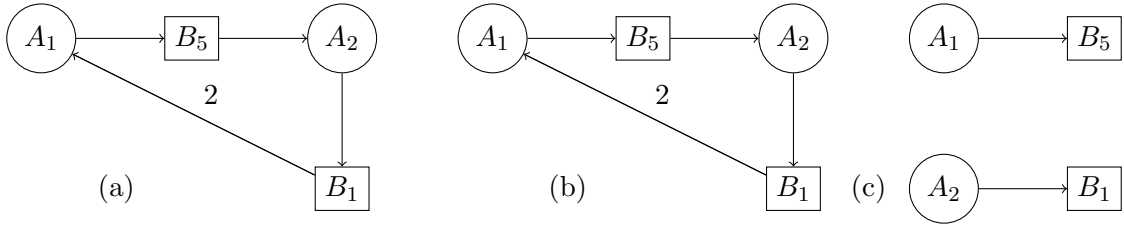

\section{Multistationarity}\label{sec:multi}

This section contains the main result of the paper.  First we will show
 that  the existence of multistationarity requires a negative term in $a_r (v,x)$
given by (\ref{gtrep-coef11}). Thus a critical fragment of order $r$
corresponding uniquely  to a negative term in $a_r (v,x) $ must be
present in the bipartite graph of a  conservative
biochemical mechanism model showing multistationary.

We have  that the coefficient $a_r(v,x)$ contains at least one
positive term  that corresponds to a product of positive diagonal
entries of the negative Jacobian  $-J(x,v)$. This is true since we
assumed (Sec.~\ref{prel}, Assumption~\ref{ass1}) that   each species participates as a
reactant (and a product)  in at least one reaction. 

\begin{theorem}\label{thm2}
    Recall the function $g_\eta$ from (\ref{eq:def_g_eta}). Let $\eta$
    and $k$ be given such that
    $g_\eta(k,\xi) = 0$, $\xi\in\Omega_\eta$
    has more that one solution. 
    Under Assumption~\ref{ass3}, if  
    all solutions are regular, then
     $a_r(v,x)$ contains a negative term.
  \end{theorem}
\begin{proof}
  Suppose not, i.e., $a_r(v,x)$ contains only positive terms. Then
  $a_r(v,x)>0$ for any 
    $k>0$ and $\xi\in\Omega_\eta$  (recall that $v\equiv v(k,x)$ and $x\equiv x(\eta,\xi)$). 
  This is in contradiction with Corollary~\ref{cor:cor2}. Therefore
  $a_r(v,x)$ contains at least one negative term.
\end{proof}

\begin{corollary}\label{cor:cor4}
    In the setting of Theorem~\ref{thm2}, 
  if the system $g_{\eta} (k,\xi)=0$, $\xi\in\Omega_\eta$ has multiple
    solutions,
  then the bipartite  graph  of the  conservative biochemical
  mechanism  (\ref{cr}) contains a critical fragment  $S_{r}$,  where
  $r$ is the rank of the stoichiometric matrix $N$.  
\end{corollary}
\begin{proof}
This follows by Theorem~\ref{thm2}, Corollary~\ref{cor:cor3} and the one-to-one correspondence between a negative term in $a_r(v,x)$ and a critical fragment.
\end{proof}

\textbf{Example.} The biochemical mechanism  (\ref{eq:br_cyc}) is conservative since its concentrations satisfy the conservation relation 
$x_1+x_2 +x_3=c_0$.   The system  (\ref{eq:de_cyc}) has no boundary equilibria since  $(0,0,0)$   does not satisfy the conservation relation $x_1+x_2 +x_3=c_0$. 
Thus   the  graph-theoretic condition developed here applies.

The existence of multistationarity requires a critical fragment of order equal to the rank of the stoichiometric matrix  by 
Corollary~\ref{cor:cor4}. Two critical fragments of order two, the rank of the stoichiometric matrix of (\ref{eq:br_cyc}), exist in  the bipartite graph of  the mechanism (\ref{eq:br_cyc}) shown in  Figure~\ref{fig:br_cyc}. The first critical fragment $S_2 \binom{1,2}{5,1}$ is shown in Figure~\ref{fig:cf_br_cyc}.  The second critical fragment $S_2 \binom{1,3}{3,1}$ is similar in structure - $S_2 \binom{1,3}{3,1}$ contains a subgraph which is a positive cycle of order 2,  $g_1=C \binom{A_1,A_3}{B_3,B_1}$ and a subgraph of edges $g_2 =\{  [A_1,B_3],[A_3,B_1]\}$.
Therefore, the existence of  multiple (always an odd number) regular  equilibria of (\ref{eq:de_cyc}) in the interior of a level set $x_1+x_2 +x_3=c_0$ for some $c_0$ is possible,  for some  values of the rate constants $k$.

\textbf{Example.}
The MAPK  network 
belongs to a family of biochemical networks known as  MAPK cascades that have 
been  extensively studied in recent years  \cite{cfr08,markevich2004signaling, fein-012}. 
 A mass action kinetics  MAPK model with a single layer is studied in \cite{fein-043}.  We use the proposed here graph-theoretic  method  to analyze  the MAPK  network for multistationarity.  We  find critical fragments  in the bipartite graph of the  MAPK network,  that are responsible for the already discovered multistationarity in \cite{fein-043,ConradiMincheva2014}.

We will   use $A$  for either a MAPKK or a MAPK, $E_1$ for
mono-phosphorylated MAPKKK or double-phosphorylated MAPKK and $E_2$
for MAPKK'ase or MAPK'ase. The biochemical mechanism involves the species $A$, $A_p$, $A_{pp}$, $E_1$, $E_2$,
$A\, E_1$, $A_p\, E_1$, $A_{pp}\, E_2$, and $A_p\, E_2$ and the 12 elementary reactions
\begin{align}\label{mapk}
\begin{split}
A + E_1 &\xrightleftharpoons[k_{2}]{k_{1}} AE_1 \xrightarrow{k_3} A_p +E_1\xrightleftharpoons[k_{5}]{k_{4}} A_p E_1\xrightarrow{k_6} A_{pp} +E_1    \\     
A _{pp}+ E_2 &\xrightleftharpoons[k_{8}]{k_{7}} A_{pp}E_2 \xrightarrow{k_9} A_p +E_2\xrightleftharpoons[k_{11}]{k_{10}} A_p E_2\xrightarrow{k_{12}} A +E_2. 
\end{split}
\end{align}

Let each species in (\ref{mapk})
be   associated with a continuously differentiable variable representing its
concentration. The concentration variables are  chosen as follows:  $x_1$ for $A$, $x_2$
for $E_1$, $x_3$ for $A\, E_1$, $x_4$ for $A_p$, $x_5$ for $A_p\,
E_1$, $x_6$ for $A_{pp}$, $x_7$ for $E_2$, $x_8$ for $A_{pp}\, E_2$
and $x_9$ for $A_p\, E_2$. The  following system of ordinary differential equations  is obtained  as a model of  
  (\ref{mapk}) with mass action kinetics
\allowdisplaybreaks{
\begin{subequations}\label{mapk-sys}
 \begin{align}
   \label{eq:ode_dis_dis_1}
   \dot{x}_1 & = -k_{1}\, x_{1}\, x_{2} + k_{2}\, x_{3} + k_{12}\,
   x_{9}        \\ 
   \dot{x}_2 &= -k_{1}\, x_{1}\, x_{2} + (k_{2}+k_{3})\, x_{3}
   -k_{4}\,  x_{2}\, x_{4} + (k_{5}+k_{6})\, x_{5}          \\
   \dot{x}_3 &= k_{1}\, x_{1}\, x_{2} - (k_{2}+k_{3})\, x_{3}      \\    
   \dot{x}_4 &= k_{3}\, x_{3} -k_{4}\, x_{2}\, x_{4} + k_{5}\,
   x_{5} + k_{9}\, x_{8} -k_{10}\, x_{4}\, x_{7} + k_{11}\, x_{9}  \\        
   \dot{x}_5 &= k_{4}\, x_{2}\, x_{4} - (k_{5}+k_{6})\, x_{5}          \\
   \dot{x}_6 &= k_{6}\, x_{5}-k_{7}\, x_{6}\, x_{7}+k_{8}\, x_{8}    \\      
   \dot{x}_7 &= -k_{7}\, x_{6}\, x_{7} + (k_{8}+k_{9})\,
   x_{8} - k_{10}\, x_{4}\, x_{7} + (k_{11}+k_{12})\, x_{9}            \\
   \dot{x}_8 &= k_{7}\, x_{6}\, x_{7} - (k_{8}+k_{9})\, x_{8}          \\
   \dot{x}_9 &= k_{10}\, x_{4}\, x_{7} - (k_{11}+k_{12})\, x_{9}   \label{eq:ode_dis_dis_9}
 \end{align}
\end{subequations}}

 Since the total concentrations of $E_1$, $E_2$ and $A$ are constant, three 
 conservation relations exist
\begin{subequations}\label{massc}
 \begin{align}
   \label{eq:con_rel_dis_dis_1}
   x_2+x_3+x_5 &= c_1,          \\
   x_7+x_8+x_9 &= c_2          \\
   \label{eq:con_rel_dis_dis_3}
   x_1+x_3+x_4+x_5+x_6+x_8+x_9 &= c_3.
 \end{align}
\end{subequations}
where each $c_i >0$, $i=1,2,3$. Thus (\ref{massc}) can be written as $\widetilde{W}^T x =c_0$ where
 \begin{equation}
   \widetilde{W}^T = 
   \left[
     \begin{array}{rrrrrrrrr}
       0&1&1&0&1&0&0&0&0 \\
       0&0&0&0&0&0&1&1&1 \\
       1&0&1&1&1&1&0&1&1 \\
            \end{array}
   \right] .
 \end{equation}
Correspondingly the level set  of the system  (\ref{mapk-sys}) is   
\begin{equation}\label{eq:tilde_w_c0}
\widetilde{\omega}_{c_0} = \{ x \geq 0 \; | \; \widetilde{W}^T x =c_0 \} . 
\end{equation}

Since  the MAPK network (\ref{mapk}) is conservative by (\ref{massc}), the theory developed here  applies  to it,
provided that the model system (\ref{mapk-sys}) does not  have any boundary equilibrium (equilibrium with at least one   zero coordinate) 
in $\tilde{\omega}_{c_0}$. The following lemma is part of Lemma~3.1 in the Supporting information of \cite{ConradiMincheva2014}.
 \begin{lemma}\label{lem0}
 The set $\widetilde{\omega}_{c_0}$  contains no boundary equilibria of  the  system (\ref{mapk-sys}).  
 \end{lemma}

The  bipartite graph of the MAPK network  (\ref{mapk}) is  shown in Figure \ref{fig:single_mapk_graph}. 
\begin{figure}[h!]
  \centering  
  \includegraphics[width=\textwidth]{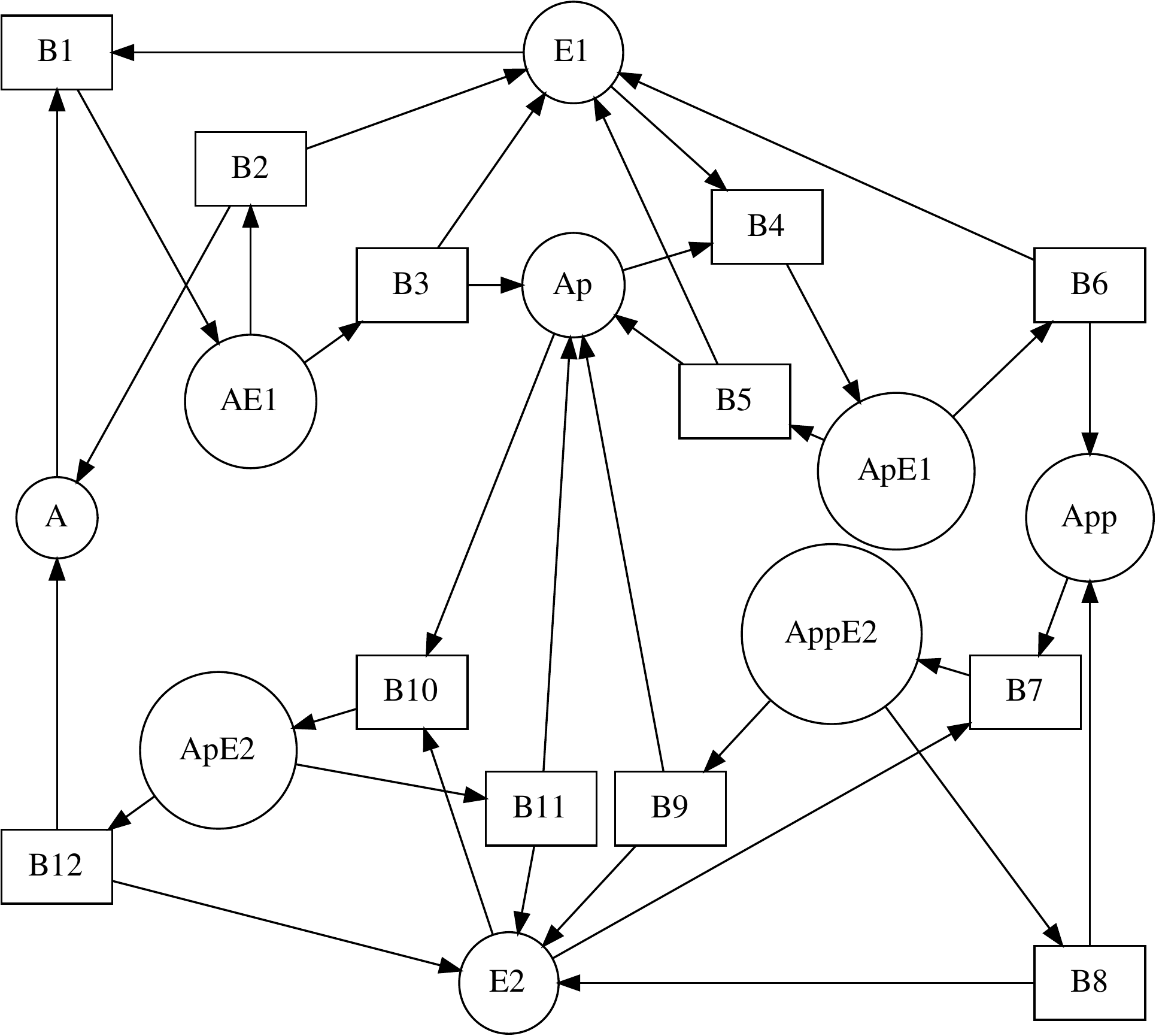}
  \caption{Bipartite digraph of the  single-layer MAPK network. Reproduced from \cite{WaltherHMincheva2014} under Open Access License Agreement.}
  \label{fig:single_mapk_graph}
\end{figure}

 The  necessary condition for multistationarity  requires  the existence of 
a critical fragment of order equal to the rank of the stoichiometric matrix by Corollary~\ref{cor:cor4}. Since the rank of the stoichiometric matrix for  the MAPK network (\ref{mapk}) equals $6$, using the package GraTeLPy   we  have enumerate all critical fragments of  order 6 in  \cite{WaltherHMincheva2014}. The 9 critical fragments of order 6 of the MAPK network  are shown in Figure \ref{fig:single_mapk_fragments_1} and Figure \ref{fig:single_mapk_fragments_2}. Therefore, the existence of  multiple (an odd number) regular  equilibria of the system (\ref{mapk-sys}) in the interior of a level set  (\ref{eq:tilde_w_c0})  for some $c_0$ is possible for some  values of the rate constants $k$.  In fact we show in \cite{ConradiMincheva2014} that for some given values of  the rate constants $k$,    three equilibria in (\ref{eq:con_rel_dis_dis_1})--(\ref{eq:con_rel_dis_dis_3})  exist where $c_0$ is chosen based on Corollary~1 in   \cite{ConradiMincheva2014}.

\begin{figure}[h!]
  \centering
  \includegraphics[width=0.9\textwidth]{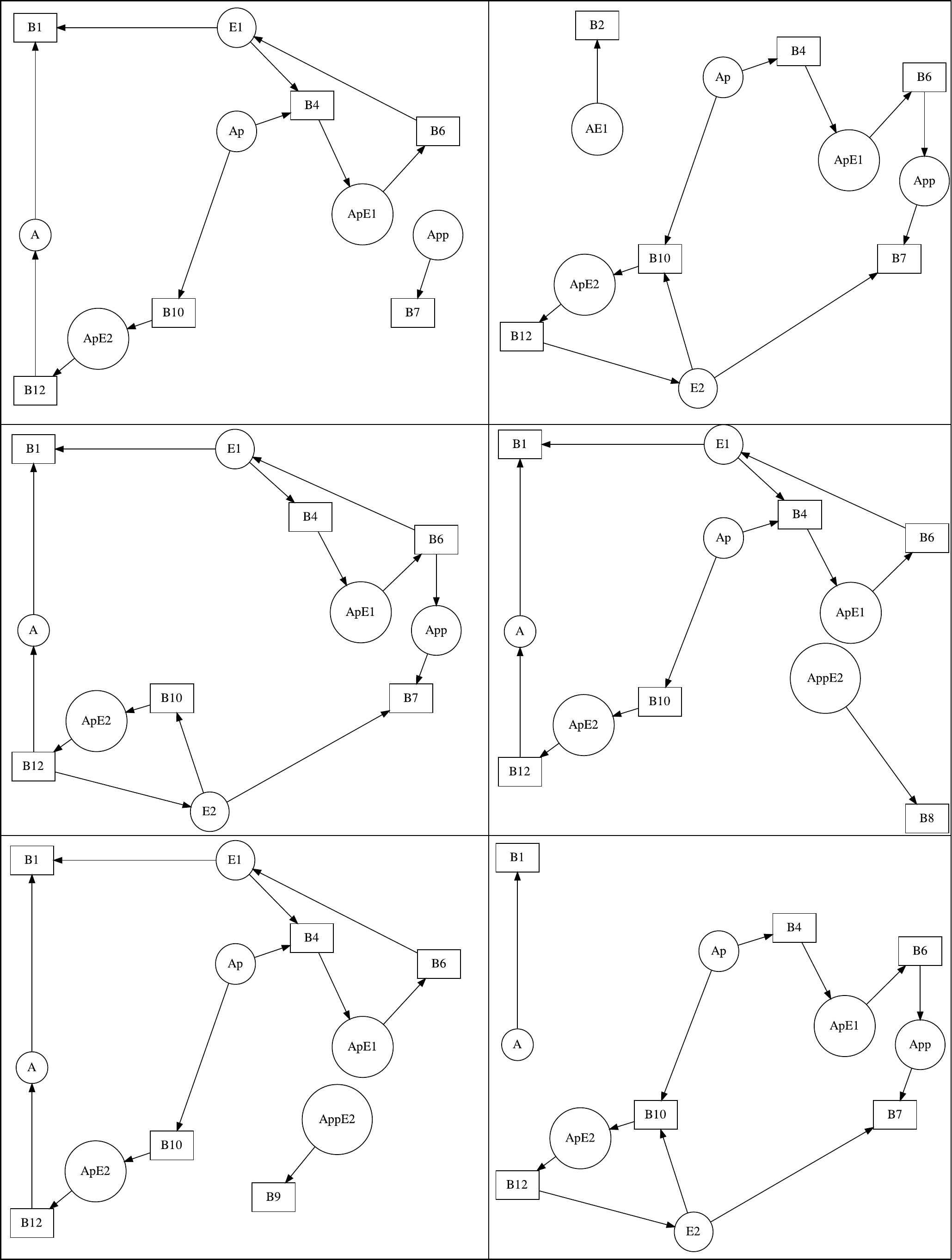}
  \caption{Critical fragments of the  single-layer MAPK network found by GraTeLPy. Reproduced from \cite{WaltherHMincheva2014} under Open Access License Agreement.}
  \label{fig:single_mapk_fragments_1}
\end{figure}

\begin{figure}[h!]
  \centering
  \includegraphics[width=\textwidth]{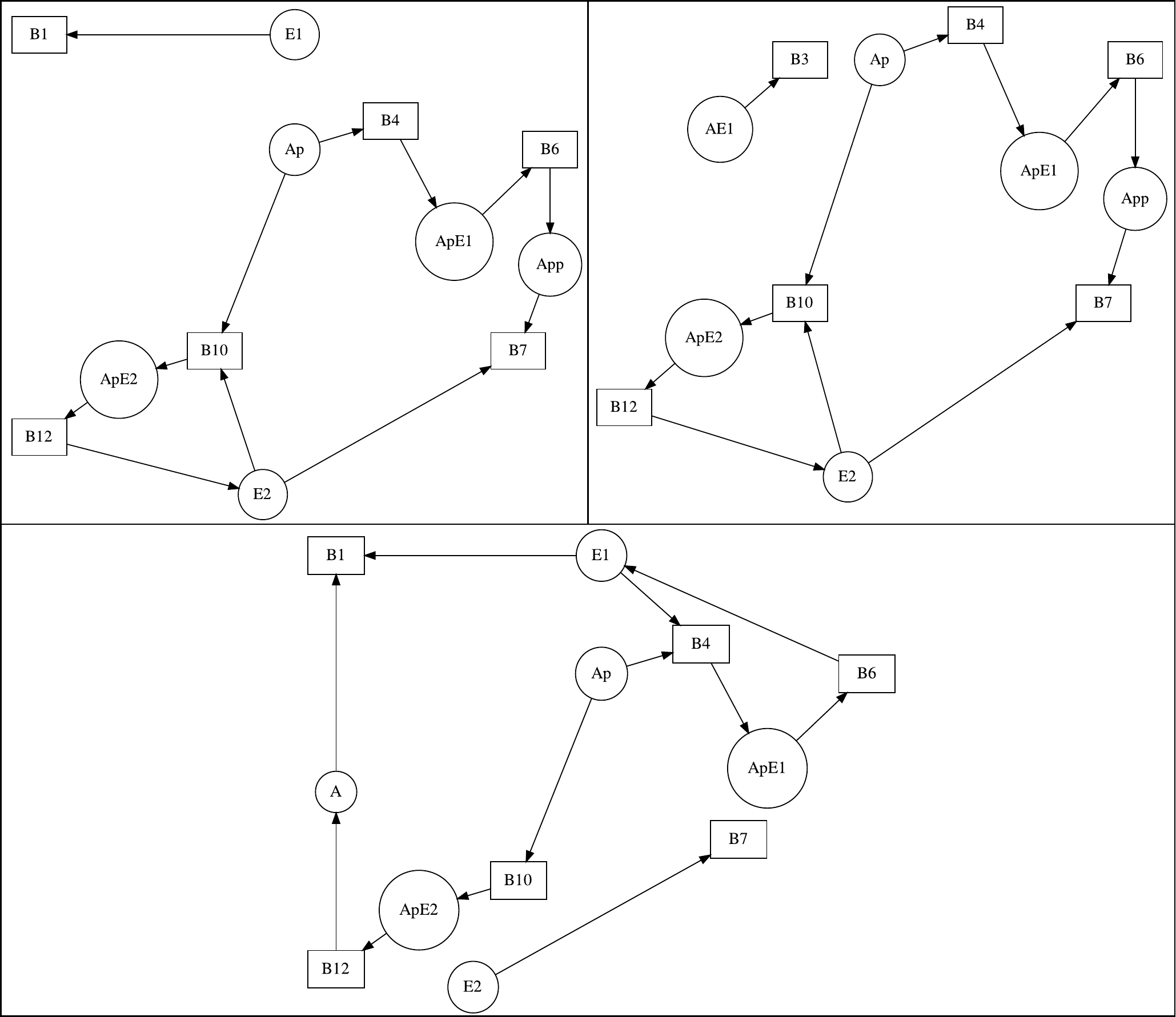}
  \caption{Critical fragments of the single-layer MAPK network found by GraTeLPy. Reproduced from \cite{WaltherHMincheva2014} under Open Access License Agreement.}
  \label{fig:single_mapk_fragments_2}
\end{figure}

\section{Discussion}
We have studied conservative biochemical mechanisms characterized by bounded concentrations of all their species. 
We have obtained a graph-theoretic condition for multistationarity for conservative biochemical mechanisms  with mass action kinetics. In essence the graph-theoretic condition is the same as for non-conservative biochemical mechanisms - the existence of a critical fragment 
of order $r$, the rank of the stoichiometric matrix,  is required for multistationarity. The difference between the case of conservative and non-conservative biochemical mechanisms is that,  in the first case we can  apply degree theory arguments \cite{Deimling2010} and in the second only  bifurcation theory can be applied. In the case of conservative biochemical mechanisms the existence of a positive equilibrium in the  level set $\omega_{c_0}$  (\ref{eq:def_om}) is always guaranteed. If multiple regular  equilibria  exist (see Sec. \ref{sec:degree_g}) in some level set  $\omega_{c_0}$ for some values of the rate constants, then the number of equilibria is  always odd. 

For large mechanisms with many species and reactions  the package GraTeLPy \cite{WaltherHMincheva2014} can be used to search for critical fragments (of order  equal to the rank of the stoichiometric matrix),  that are necessary for multistationarity.

Other related graph-theoretic conditions for  multistationarity  have been developed recently. 
In the work of Craciun and Feinberg the undirected   species-reaction (SR) graph is used and a  graph-theoretic condition that precludes multistationarity in open system  mass-action kinetics models   for any parameter values  is obtained  \cite{Craciun2006}. Banaji and Craciun  obtain graph-theoretic conditions  for injectivity and uniqueness of  equilibria regardless of parameter values  in chemical kinetics models  using the SR graph \cite{BanajiCraciun 2010}. In an earlier work the same authors use a signed, directed, labeled, bipartite multigraph, termed the ``DSR graph'' to obtain a graph-theoretic condition 
that rules out multiple equilibria of general interaction  networks models \cite{BanajiCraciun2009}. 

 In  \cite{hofbauer1990} degree theory is used to study the  number of equilibria   of ecological differential equations where $\dot{x}_i =x_i f_i(x)$ for all $i$.
Degree theory methods have also been used in \cite{CraciunHelton2008} to  determine the number of equilibria for complex biochemical reaction networks. Degree theory arguments are used to find parameter values (rate constants and total concentrations) such that  the MAPK model (\ref{mapk-sys})   has three equilibria or a single equilibria in  \cite{ConradiMincheva2014}.

\bibliographystyle{plain}
\bibliography{graph-degree}

\end{document}